\documentclass{llncs}
\usepackage{tikz}
\usepackage{xcolor}
\usepackage[T1]{fontenc} 

\usepackage{amsfonts,amssymb,amsmath}
\usepackage{algorithm}
\usepackage[noend]{algpseudocode}
\usepackage{ifthen}

\newtheorem{observation}[lemma]{Observation}

\begin{document}

\title{Space-efficient conversions from SLPs}
\author{Travis Gagie\inst{1,5},  Adri\'an Goga\inst{2}, Artur Je\.z\inst{3}, Gonzalo Navarro\inst{4,5}}

\institute{
Faculty of Computer Science, Dalhousie University, Canada
\and
Department of Computer Science, Comenius University in Bratislava, Slovakia
\and
Institute of Computer Science, University of Wrocław, Poland
\and
Department of Computer Science, University of Chile, Chile
\and
CeBiB — Center for Biotechnology and Bioengineering, Chile
}

\maketitle

\begin{abstract} 
We give algorithms that, given a straight-line program (SLP) with $g$ rules that generates (only) a text $T [1..n]$, builds within $O(g)$ space the Lempel-Ziv (LZ) parse of $T$ (of $z$ phrases) in time $O(n\log^2 n)$ or in time $O(gz\log^2(n/z))$. We also show how to build a locally consistent grammar (LCG) of optimal size $g_{lc} = O(\delta\log\frac{n}{\delta})$ from the SLP within $O(g+g_{lc})$ space and in $O(n\log g)$ time, where $\delta$ is the substring complexity measure of $T$. Finally, we show how to build the LZ parse of $T$ from such a LCG within $O(g_{lc})$ space and in time $O(z\log^2 n \log^2(n/z))$. All our results hold with high probability.
\end{abstract}

\section{Introduction}

With the rise of enormous and highly repetitive text collections \cite{Navacmcs20.3}, it is becoming practical, and even necessary, to maintain the collections compressed all the time. This requires being able to perform all the needed computations, like text searching and mining, directly on the compressed data, without ever decompressing it.

As an example, consider the modest (for today's standards) genomic repository {\em 1000 Genomes} \cite{DAB15} containing the genomes of 2,500 individuals. At the typical rate of about 3 billion bases each, the collection would occupy about 7 terabytes. Recent projects like the {\em Million Genome Initiative}\footnote{\tt https://digital-strategy.ec.europa.eu/en/policies/1-million-genomes} would then require petabytes. The 1000 Genomes project stores and distributes its data already in compressed form\footnote{In VCF, \tt https://github.com/samtools/hts-specs/blob/master/VCFv4.3.pdf} to exploit the fact that, compared to a reference genome, each individual genome has only one difference every roughly 500 bases, on average. Certainly one would like to manipulate even such a modest collection always in compressed form, using gigabytes instead of terabytes of memory!

Some compression formats are more useful for some tasks than others, however. For example, Lempel-Ziv compression \cite{LZ76}
tends to achieve the best compression ratios, which makes it more useful for storage and transmission.
Grammar compression \cite{KY00} yields slightly larger files, but in exchange it can produce $T$ in streaming form, and provide direct access to any text snippet \cite{BLRSRW15}, as well as indexed searches \cite{CNPjcss21}.
Locally consistent grammars provide faster searches, and support more complex queries, while still being bounded by well-known repetitiveness measures \cite{CEKNPtalg20,KNPtit22,KNOlatin22,kempa2023collapsing,gawrychowski2018optimal}. The run-length-encoded Burrows-Wheeler Transform of $T$ requires even more space \cite{KK20}, but in exchange it enables full suffix tree functionality \cite{GNPjacm19}.

It is of interest, then, to {\em convert} from one format to another. Doing this conversion by decompressing the current format and then compressing to the new one is impractical, as it is bound to use $\Omega(n)$ space, which in practice implies running $\Theta(n)$-time algorithms on secondary storage. Our goal is to perform the conversions {\em within space proportional to the sum of the input and output files}. We call this a {\em compressed conversion}.

Let $z$, $g$, $g_{lc}$, and $r$ be the asymptotic (i.e., up to constant factors) sizes of the Lempel-Ziv (LZ) parse of a text string $T[1..n]$, a straight-line program (SLP) or context-free grammar that expands to $T$, a locally consistent grammar (LCG) that expands to $T$, and the number of runs in the Burrows-Wheeler Transform (RLBWT) \cite{BW94} of $T$, respectively.
We refer to some SLP or LCG because finding the smallest SLP (and probably LCG) that generates a given string is NP-complete \cite{CLLPPSS05}.
It holds that $z \le g \le g_{lc}$ and, in practice, $g_{lc} \le r$.

Rytter~\cite{Ryt03} showed how to build an SLP of size $g = O(z\log(n/z))$ within $O(g)$ space and time from the LZ parse of $T$ in the non-overlapping case (i.e., when phrases cannot overlap their sources), and Gawrychowski~\cite[Lemma 8]{gawrychowski2011pattern}
extended this result to general LZ.
Note that these algorithms take space {\em and time} proportional to the compressed input and output. {\em fully-compressed conversions} those that take time polynomial in the size of the input, the output, and $\log n$. On highly repetitive texts, all the given measures can be exponentially smaller than $n$, hence the relevance of fully-compressed conversions.

Policriti and Prezza \cite{PP17} showed how to convert 
from RLBWT to LZ in $O(r+z)$ space and $O(n\log r)$ time, and back in the same space and $O(n\log(rz))$ time.
Kempa and Kociumaka \cite{KK20} convert from LZ to RLBWT in $O(z\log^8 n)$; this conversion is then fully-compressed. 
Arimira et al. \cite{arimura2023optimally} recently showed how to convert from the CDAWG of size $e$ to either RLBWT or LZ, both in $O(e)$ time and space, though $e$ is the weakest among the commonly accepted repetitiveness measures \cite{navarro2021indexing}. 

For brevity we omit a large body of work on producing compressed representations from the original string $S$, aiming to use little space on top of $S$ itself. We also omit work on compression formats that are too weak for repetitive data, like LZ78 or run-length compression of the text.


In this paper we contribute to the state of the art with compressed and fully-compressed conversions between various formats, all of which then use space linear in the input plus the output, and work correctly with high probability:
\begin{enumerate}
    \item A compressed conversion from any SLP to LZ in $O(n\log^2 n)$ time.
    \item A fully-compressed conversion from any SLP to LZ in $O(gz\log^2(n/z))$ time.
    \item A compressed conversion from any SLP to a certain (particularly small) LCG \cite{christiansen2020optimal} in $O(n\log g_{lc})$ time.
    \item A fully-compressed conversion from LCGs of some particular kind \cite{christiansen2020optimal,KNOlatin22} to LZ in $O(z\log^2 n \log^2(n/z))$ time.
\end{enumerate}
The third conversion builds a particular LCG whose size is the optimal $O(\delta\log\frac{n}{\delta})$ 
\cite{KNOlatin22}; other similar LCSLPs \cite{CEKNPtalg20} can be produced analogously. Our contributions together with previously known conversions are depicted Fig. \ref{fig:conversions_diagram}.

\begin{figure}[ht]
    \centering
    \begin{tikzpicture}[align=center, node distance=3cm]
\node (g) {$g$};
\node[right of=g, above of=g] (glc) {$g_{lc}$};
\node[right of=glc, below of=glc] (z) {$z$};
\node[right of=z, below of=z] (e) {$e$};
\node[right of=e, above of=e] (r) {$r$};

\path [->, very thick] (g) edge node[above] {$O (gz \log^2 (n/z))$} (z);
\path [->, dashed, very thick] (g) edge[bend right] node[above] {$O (n\log^2 n)$} (z);
\path [->, dashed, very thick] (g) edge node[left] {$O (n\log g_{lc})$} (glc);
\path [->, very thick] (glc) edge node[right] {$O(z \log^2 n \log^2 (n/z))$} (z);

\path [->] (z) edge[bend right] node[above] {$O(z\log(n/z))$ \cite{Ryt03,gawrychowski2011pattern}} (g);
\path [->, dashed] (z) edge node[above] {$O(n\log(rz))$ \cite{PP17}} (r);
\path [->] (z) edge[bend left] node[above] {$O(z\log^ 8 n)$ \cite{KK20}} (r);

\path [->, dashed] (r) edge[bend left] node[above=0.1cm] {$O(n\log r)$ \cite{PP17}} (z);

\path [->] (e) edge node[left=0.1cm] {$O(e)$ \cite{arimura2023optimally}} (z);
\path [->] (e) edge node[right] {$O(e)$ \cite{arimura2023optimally}} (r);
\end{tikzpicture}
    \caption{A diagram capturing all known conversions between compression methods of highly repetitive data. Dashed lines represent compressed conversions, while solid lines represent fully-compressed ones. Thick lines mark contributions of this article.}
    \label{fig:conversions_diagram}
\end{figure}
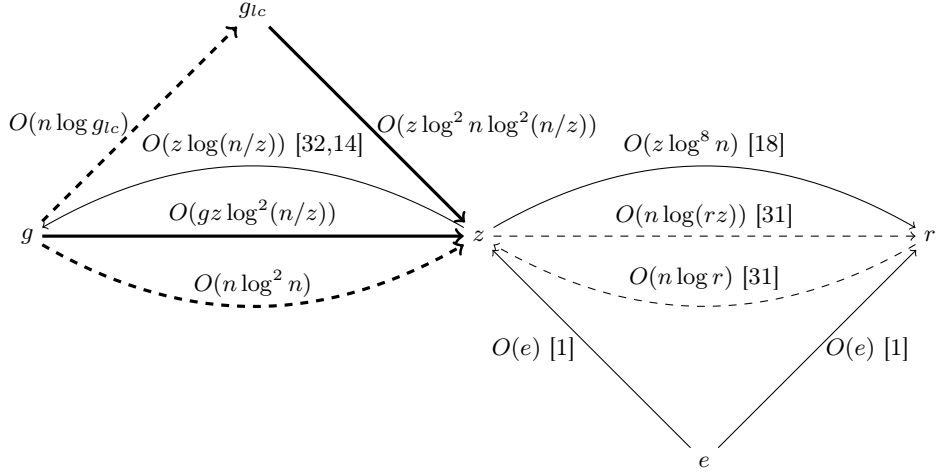

\section{Preliminaries}

A {\em string} $T[1..n]$ is a sequence of symbols $T[1]\,T[2]\ldots T[n]$ over an ordered alphabet $\Sigma$. For every $1 \le i,j \le n$, $T[1..i] = T[..i]$ is a {\em prefix} of $T$, $T[j..n] = T[j..]$ is a {\em suffix} of $T$, and $T[i..j]$ is a {\em substring} of $T$, which is the empty string $\varepsilon$ if $i>j$. The {\em length} of $T[1..n]$ is $|T| = n$; the length of $\varepsilon$ is $|\varepsilon|=0$. The {\em concatenation} of two strings $S \cdot S'$ is defined as $S[1]\,S[2]\ldots S[|S|]\,S'[1]\,S'[2]\ldots S'[|S']]$. The {\em lexicographic} order between strings $S \not= S'$ is defined as that between $S[1]$ and $S'[1]$ if these are different, or as the lexicographic order between $S[2..]$ and $S'[2..]$ otherwise; the empty string $\varepsilon$ is smaller than every other string. The {\em co-lexicographic} order is defined as the lexicographic order between the reversed strings.

The \textit{Karp-Rabin fingerprint} or the \textit{Karp-Rabin hash} of a string $S[1..n]$ is a value $\phi(S) = \sum_{i=1}^n \left( S[i] x^i \right) \mod p$, for a prime $p$ and $x < p$ \cite{karp1987efficient}. The crucial property of this hash is that if $X \neq Y$, then $\phi(X) \neq \phi(Y)$ with high probability. Another well-known and useful property is that for strings $S, S', S'$ for which $S = S'\cdot S''$ holds, we can compute the hash of any of the strings knowing the hashes of the other two, in $O(1)$ time (see, e.g., \cite{NP18}).

A \textit{straight-line program (SLP)} of a text $T$ is a context-free grammar in Chomsky's normal form
(so in particular: each rule is at most binary)
generating only $T$, which contains exactly one rule for each nonterminal and the rules can be linearly ordered, such that for any rule $X \to YZ$ it holds that the rules for both $Y$ and $Z$ precede the rule for $X$ in the ordering.

The height of the SLP is the height of the derivation tree,
i.e.\ the height of a letter is $0$ and the height of a nonterminal $X$ with a (unique) rule $X \to YZ$ is $1$ plus maximum of height of $Y$ and $Z$.
The size of the SLP is the number of its rules. We define the {\em expansion} of a nonterminal $X$ as the string it produces: $\exp(a)=a$ if $a$ is a terminal symbol, and $\exp(X) = \exp(Y) \cdot \exp(Z)$ if $X \to YZ$.

We say that a grammar is a \textit{locally consistent grammar (LCG)} if it is constructed by iteratively applying rounds of a particular locally consistent parsing, which guarantees that matching fragments $S[i..j] = S[i'..j']$ are parsed the same way, apart from the $O(1)$ blocks from either end. This key property is lifted to such grammars, for which matching fragments are spanned by almost identical subtrees of the parse tree, differing in at most $O(1)$ flanking nonterminals at each level \cite{gawrychowski2018optimal,christiansen2020optimal}.
Such a parsing is defined in Section~\ref{sec:locally_consistent_parsing}.

The Lempel-Ziv (LZ) parse of a string $T$ \cite{LZ76} is a sequence $F_1, F_2, \ldots, F_z$ of phrases,
such that $F_1 \cdot F_2 \cdots F_z = T[1..n]$ and $F_i$ is either a single letter,
when this letter is not present in $F_1 \cdot F_2 \cdots F_{i-1}$, or else $F_i$
is the maximal string that occurs twice in $F_1 \cdot F_2 \cdots F_i$,
that is, it has an occurrence starting within $F_1 \cdot F_2 \cdots F_{i-1}$;
in non-overlapping LZ we additionally require that $F_i$ occurs within $F_1 \cdot F_2 \cdots F_{i-1}$.
It is known that $z \le g = O(z\log(n/z))$, where $g$ is  
the size of the smallest grammar generating $T$~\cite{Ryt03,CLLPPSS05,gawrychowski2011pattern}. 

We assume the standard word-RAM model of computation with word length $\Theta (\log n)$, in which basic operations over a single word take constant time. Some of our results hold with high probability (whp), meaning with probability over $1-n^{-c}$
for any desired constant $c$.
We can make the constant arbitrarily large at the cost of increasing the constant multiplying the running time.

\section{Building the LZ parse from an SLP in $\tilde{O}(n)$ time}
\label{sec:lz77npolylog}

Our first result computes the LZ parse of a text $T[1..n]$ given an arbitrary SLP of size $g$ that represents $T$, in time $O(n \log^2 n)$ and space $O(g)$;
note that the classic LZ constructions use suffix trees or arrays and use $\Omega(n)$ space.
We first describe a couple of tools we need to build on the SLP before doing the conversion.

\begin{lemma}
\label{lem:basic_checks}
    Given an SLP of size $g$ for $T[1..n]$ we can construct in $O(g)$ time and space a new SLP $G$, and augment it with a data structure such that:
    \begin{itemize}
        \item $G$ has height $O(\log n)$.
        \item Any $T [i]$ can be accessed in $O(\log n)$ time.
        \item The Karp-Rabin fingerprint of any $T [i..j]$ can be computed in $O(\log n)$ time. 
        \item The longest common prefix of any $T [i..j]$ and $T [i'..j']$ can be computed (whp) in $O(\log^2 n)$ time. 
        \item Any $T [i..j]$ and $T [i'..j']$ can be compared lexicographically and co-lexicograph\-ically (whp) in $O(\log^2 n)$ time. 
    \end{itemize}
\end{lemma}

\begin{proof}
Assume that we are given an SLP with $g$ rules for a text $T [1..n]$.
Ganardi et al.~\cite{GJL21} showed how we can, in $O(g)$ time and space, turn it into an SLP $G$ of size $O(g)$ and height $O(\log n)$. We work on $G$ from now on. Given its height, it is easy to augment $G$ with $O(g)$-space structures that, in $O( \log n)$ time,
finds any character $T [i]$ and returns the Karp-Rabin hash of any substring $T [i..j]$ (see Ganardi et al.~\cite{GJL21}, which refers to a simple data structure from Bille et al.~\cite{BGCSVV17}).
Given two substrings of $T$, we can then compute their longest common prefix in $O(\log^2 n)$ time---whp of obtaining the correct answer---by exponentially searching for its length $\ell$ \cite[Thm.~3]{BGCSVV17};
by checking their characters at offset $\ell+1$ we can also 
compare the substrings of $T$ lexicographically within the same time complexity.
We can similarly compute the longest common suffix of two substrings and thus compare them co-lexicographically (by comparing the preceding characters). \qed
\end{proof}

We will also use a variant of a z-fast trie. 

\begin{lemma} \label{lem:zfast}
Let $\mathcal{S}$ be a sorted multiset of $m$ strings of total length $n$. Then one can build, in $O(m)$ time whp, a data structure of size $O(m)$ that, given a string $P$, finds in $O(f_h \log n)$ time the lexicographic range of the strings in $\mathcal{S}$ prefixed by $P$, where $f_h$ is the time to compute a Karp-Rabin fingerprint of a substring of $P$. If the range is empty, the answer can be incorrect. If the range is nonempty, the answer is correct whp.
\end{lemma}
\begin{proof}
The structure is the z-fast trie of Belazzougui et al.~\cite[Thm.~5]{BBPV10}. A simpler construction was given by Kempa and Kosolobov \cite{KK17}, and it was then fixed, and its construction analyzed, by Navarro and Prezza \cite[Sec.~4.3]{NP18}.
\end{proof}

We will resort to a classic grammar-based indexing method \cite{CNPjcss21}, for which we need a few definitions and properties.

\begin{definition} \label{def:slpindex}
The {\em grammar tree} of an SLP $G$ is formed by pruning the parse tree, converting to leaves, for every nonterminal $X$, all the nodes labeled $X$ but the leftmost one.
An occurrence of a string $P$ in $T$ is called {\em primary} if it spans more than one leaf in the grammar tree; otherwise it is contained in the expansion of a leaf and is called {\em secondary.} If a primary occurrence of $P$ occurs in $\exp(X)$, with rule $X \to YZ$, starting within $\exp(Y)$ and ending within $\exp(Z)$, we say that the position $P[j]$ that aligns to the last position of $\exp(Y)$ is the {\em splitting point} of the occurrence.
\end{definition}

A small exception to this definition is that, if $|P|=1$, we say that its primary occurrences are those where it appears at the end of $\exp(X)$ in any leaf $X$ of the grammar tree. We now give a couple of results on primary occurrences.

\begin{lemma}[\cite{CNPjcss21}] \label{lem:primary_some}
Every pattern $P$ that occurs in $T$ has at least one primary occurrence. 
\end{lemma}

\begin{observation} \label{obs:leftmost}
If $X$ is the lowest nonterminal containing a primary occurrence of $P$ with splitting point $j$, then, by the way we form the grammar tree, this is the leftmost occurrence of $P$ under $X$ with splitting position $j$.
\end{observation}

The index sorts all rules $X \to YZ$ twice: once by the lexicographical order of $\exp(Z)$, while collecting those expansions in a multiset $\mathcal{Z}$, and once by the co-lexicographical order of $\exp(Y)$, while collecting the reversed expansions in a multiset $\mathcal{Y}$. It builds separate z-fast tries (Lemma~\ref{lem:zfast}) on $\mathcal{Y}$ and $\mathcal{Z}$, and creates a discrete $g \times g$ grid $\mathcal{G}$, where 
the cell $(x, y)$ stores the position $p$ iff the $x$th rule $X \to YZ$ in the first order is the $y$th rule in the second order, and $T[p]$ is aligned to the last symbol of $\exp(Y)$ within the occurrence of $X$ as an internal node in the grammar tree. The grid supports orthogonal range queries. The key idea of the index is that, given a search pattern $P$, for every $1 \le j \le |P|$, the lexicographic range $[y_1,y_2]$ of $P[j+1..]$ in $\mathcal{Z}$ and the lexicographic range $[x_1,x_2]$ of the reverse of $P[..j]$ in $\mathcal{Y}$, satisfy that there is a point in the range $[x_1,x_2] \times [y_1,y_2]$ of $\mathcal{G}$ per primary occurrence of $P$ in $T$ with splitting point $P[j]$. The structure $\mathcal{G}$ can determine if the area is empty, or else return a point in it, in time $O(\log g)$.
We now build our first tool towards our goal.

\begin{lemma}
    \label{lem:primary}
    Given an SLP of size $g$ generating string $T[1..n]$ we can, in space $O(g)$ and time
    $O (n)$ or $O(g \log^2 n \log g)$,
    construct a (whp correct) data structure that,
    given $1 \leq i \leq j \leq k \le n$, in $O (\log^2 n)$ time and whp of obtaining the correct answer, finds the leftmost occurrence of $T[i..k]$ in $T$ that is a primary occurrence with splitting point $T[j]$.
\end{lemma}
\begin{proof}
We build the components $\mathcal{Y}$, $\mathcal{Z}$, and $\mathcal{G}$ of the described index, correctly whp: we sort the sets $\mathcal{Y}$ and $\mathcal{Z}$ in $O(g)$ space and $O(n)$ time, we build the z-fast tries in $O(g)$ space and time, and we build the grid data structure in $O(g)$ space and $O(g\sqrt{\log g})$ time (see Navarro and Prezza \cite[Sec.~4.3]{NP18} for all those time an space complexities). We note that, by using Lemma~\ref{lem:basic_checks}, we can also do the sorting correctly whp in $O(g)$ space and $O(g\log g \cdot \log^2 n)$ time.

We use those structures to search for $P = T[i..k]$ with splitting point $T[j]$, that is, we search the z-fast trie of $\mathcal{Z}$ for $T[j+1..k]$ and the z-fast trie of $\mathcal{Y}$ for $T[i..j]$ reversed, in time $O(f_h \log n)$; recall Lemma~\ref{lem:zfast}. Since the substrings of $T[i..j]$ are also substrings of $T$, we can compute the Karp-Rabin hash of any substring of $T[i..j]$ in time $f_h=O(\log n)$ by Lemma~\ref{lem:basic_checks}, so this first part of the search takes time $O(\log^2 n)$. Recall from Lemma~\ref{lem:zfast} that this search can yield incorrect results (only if the ranges sought are empty, whp).

We now use $\mathcal{G}$ to determine if there are points in the corresponding area. If there are none, then whp $T[i..k]$ does not occur in $T$ with splitting point $T[j]$. If there are some, then we obtain the value $p$ associated with any point in the range, and compare the Karp-Rabin hash of $T[p-(j-i)..p+(k-j)]$ with that of $T[i..k]$. If they differ, then $T[i..k]$ has no occurrences with splitting point $T[j]$; otherwise whp the z-fast tries gave the correct range and there are occurrences. This check takes $O(\log n)$ time.

Once we know that (whp) there are occurrences with splitting point $T[j]$, we want the leftmost one. Each point within the grid range may correspond to a different rule $X \to YZ$ that splits $T[i..k]$ at $T[j]$; therefore, by Observation~\ref{obs:leftmost}, we want the minimum of the $p$ values stored for the points within the range.
This kind of two-dimensional range minimum query can be solved in time $O(\log^{1+\varepsilon} g)$ and $O(g)$ space, for any constant $\varepsilon>0$, with an enhancement of $\mathcal{G}$ that uses $O(g)$ space and can be built in time $O(g\log g)$ \cite{navarro2014wavelet,Cha88}. This completes the query. \qed
\end{proof}

Finally, we will need the following observation on the monotonicity of occurrences in $T$, even when we stick to some splitting point.

\begin{observation} \label{obs:prefix}
If $P$ has a primary occurrence in $T$ with splitting point $P[j]$, then any prefix $P'=P[..k]$, for any $j < k < |P|$, also has a primary occurrence with splitting point $P'[j]$.
\end{observation}
\begin{proof}
Let the primary occurrence of $P$ span $Y$ and $Z$ for a rule $X \to YZ$, that is, $P$ appears in $\exp(X)$ and the occurrence starts in $\exp(Y)$ and ends in $\exp(Z)$, with $P[j]$ aligned to the last position of $\exp(Y)$. Then it follows that $P'=P[..k]$ satisfies the same conditions, so a primary occurrence of $P'$ with splitting point $P'[j]$ starts at the same text position.\qed
\end{proof}

\begin{algorithm}[t]
  \begin{algorithmic}[1]
    \State $i \gets 1$
    \While {$i \le n$}
        \If{$T[i]$ does not occur in $T[1..i-1]$}
            \State \textbf{output} literal phrase with character $T[i]$
            \State $i \gets i+1$
        \Else
            \State $j \gets i$, $k \gets i-1$
            \Repeat 
                \State $k \gets k+1$
                \State // invariant: $T[i..k]$ occurs to the left of $i$ with splitting point $T[j]$ 
                \State // \hspace{1.4cm} and not with any in $T[i..j-1]$
                \While{$k < n \textbf{ and } \mathrm{leftmost}(i,j,k+1) < i$} 
                    $k \gets k+1$ 
                \EndWhile
                \State $t \gets \mathrm{leftmost}(i,j,k)$
                \If{$k < n$}
                    \While{$j \le k+1 \textbf{ and } \mathrm{leftmost}(i,j,k+1) \le i$}
                        $j \gets j+1$
                    \EndWhile
                \Else $~j \gets n+1$
                \EndIf
                \State // invariant holds for $k+1$ if $k<n$
                \State // if $j>k$ then $T[i..k+1]$ does not occur to the left of $i$
            \Until{$j > k$}
            \State \textbf{output} phrase for $T[i..k]$ with source starting at $T[t..]$
            \State $i \gets k+1$
        \EndIf
    \EndWhile
  \end{algorithmic}
\caption{Obtaining the LZ parse of $T[1..n]$ from an SLP generating $T$. Function $\mathrm{leftmost}(i,j,k+1)$ uses Lemma~\ref{lem:primary} to return the leftmost occurrence of $T[i..k]$ in $T$ with splitting point $T[j]$, or $n+1$ if not found.}
\label{alg:lzparse}
\end{algorithm}

We are now ready to give the final result. Algorithm~\ref{alg:lzparse} gives the pseudocode.

\begin{theorem}
\label{thm:npolylogn}
Given an SLP with $g$ rules for a text $T [1..n]$, whp we can build the LZ parse of $T$ in $O (n \log^2 n)$ time and within $O (g)$ space.
\end{theorem}
\begin{proof}
We first build the data structures of Lemma~\ref{lem:primary} in $O(n)$ time and $O(g)$ space, correctly whp. We then carry out the LZ parse by sliding three pointers left-to-right across $T$, $i \le j \le k$, as follows:
suppose that the parse for $T [1..i-1]$ is already constructed, so a new phrase must start at $i$.
We first check whether $T [i]$ appeared already in $T [1..i-1]$,\footnote{This is easily done in $O(1)$ time and $|\Sigma| \in O(g)$ space by just storing an array with the leftmost occurrence of every distinct symbol in $T$. This array is built in $O(g)$ time from the leaves of the grammar tree.}
if not then we create a one-letter phrase and proceed to $i+1$.

If $T[i]$ has occurred earlier, we start the main process of building the next phrase. The invariant is that we have found $T[i..k]$ starting before $i$ in $T$ with splitting point $T[j]$, and there is no primary occurrence of $T[i..k]$ (nor of $T[i..k']$ for any $k'>k$, by Observation~\ref{obs:prefix}) with a splitting point in $T[i..j-1]$. To establish the invariant, we initialize $j$ to $i$ and try $k$ from $i$ onwards, using Lemma~\ref{lem:primary} and advancing $k$ as long as the leftmost occurrence of $T[i..k]$ with splitting point $T[i]$ starts to the left of $i$.

Note that we will succeed the first time, for $k=i$. We continue until we reach $k=n$ (and output $T[i..n]$ as the last phrase of the LZ parse) or we cannot find $T[i..k+1]$ starting before $i$ with splitting point $T[i]$. We then try successive values of $j$, from $i+1$ onwards, using Lemma~\ref{lem:primary} to find $T[i..k+1]$ starting before $i$ with splitting point $T[j]$. If we finally succeed for some $j \le k$, we reestablish the invariant by increasing $k$ and return to the first loop, which again increases $k$ with fixed $j$, and so on.

When $j$ reaches $k+1$, it follows that $T[i..k]$ occurs before $i$ and $T[i..k+1]$ does not, with any possible splitting point. The next phrase is then $T[i..k]$, which we output, reset $i = k+1$, and resume the parsing.

Since $j$ and $k$ never decrease along the process, we use Lemma~\ref{lem:primary} $O(n)$ times, for a total time of $O(n\log^2 n)$ to build the LZ parse.\qed
\end{proof}

\section{Building the LZ parse from an SLP in $\tilde{O}(g z)$ time}
\label{sec:lz77gzlogn}
If $T$ is highly compressible, the running time $O(n\log^2 n)$ in Theorem~\ref{thm:npolylogn} could be exponential in the size $O (g)$ of the input.  We can build the parse in $O (g z \log^2 \frac n z) \subset \mathrm{poly} (g)$ time by using, instead of the machinery of the preceding section, Je\.z's~\cite{jez2015faster} algorithm for fully-compressed pattern matching. We will only balance the SLP if needed \cite{GJL21} so that its height is $O(\log n)$. We start by reminding some tools.

\begin{lemma}[\cite{Ryt03}] \label{lem:avl}
Given an SLP of height $h$ for $T$, we can in $O(h)$ time and space produce an SLP of size $O(h)$ for any desired substring $T[i..j]$ (without modifying the SLP of $T$).
\end{lemma}

Note that the SLP constructed in the Lemma above may use some of the nonterminals of the original SLP for $T$, i.e.\ its size is in principal $g + O(h)$.

\begin{lemma}[{\cite{jez2015faster}}] \label{lem:fcpm}
If $T$ and $P$ have SLPs of size $g$ and $g'$, then we can find the leftmost occurrence of $P$ in $T$ in time $O((g+g')\log|P|)$, within $O(g+g')$ space.
\end{lemma}

Assume again we have already parsed $T[1..i-1]$, and aim to find the next phrase, $T[i..k]$. We will exponentially search for $k$ using $O(\log(k-i))$ steps. Each step implies determining whether some $T[i..j]$ occurs in $T$ starting to the left of $i$ (so that $k$ is the maximum such $j$). To do this we exploit the fact that our SLP is of height $h = O(\log n$) and use Lemma~\ref{lem:avl} to extract an SLP for $T[i..j]$, of size $g' \le g + O(h)= g + O(\log n)$, in $O(h) = O(\log n)$ time\footnote{Rytter \cite{Ryt03} rebalances the grammar he extracts, but we do not need to do this.}. We then search for the SLP of $T[i..j]$ in the SLP of size $g$ of $T$ using Lemma~\ref{lem:fcpm}, in time $O((g+g') \log (j-i)) \subseteq O(g\log(k-i))$ (because $g' = \subseteq O(g)$, as $g$ is always $\Omega(\log n)$). By comparing the leftmost occurrence position with $i$ we drive the exponential search, finding $k$ in time $O (g \log^2 (k - i))$ and space $O (g)$.

Repeating this for each LZ phrase we get
$\sum_{i=1}^z g \log^2 n_i$, where $n_1, n_2, \ldots, n_z$ denote the consecutive phrase lengths.
By Jensen's inequality (since $\log^2(\cdot)$ is concave), the sum is maximized when all $n_i = n/z$.

\begin{theorem}
\label{thm:SLP_to_LZ77}
Given a straight-line program with $g$ rules for a text $T [1..n]$ whose LZ parse consists of $z$ phrases, we can build that parse in $O (g z \log^2 (n/z))$ time and $O (g)$ space.
\end{theorem}



\section{Building an LCP from an SLP in $\tilde{O}(n)$ time}
\label{sec:locally_consistent_parsing}

Locally consistent grammars (LCGs) are actually run-length context-free grammars, that is, they allow rules $X \to Y_1\cdots Y_t$ (of size $t$) and run-length rules
of the form $X \rightarrow Y^t$, equivalent to $X \rightarrow Y \cdots Y$ ($t$
copies of $Y$), of size 2.
A particular kind of LCG can be obtained from $T$ with the following procedure~\cite{KNOlatin22}. First, define 
$\ell_k = (4/3)^{\lceil k/2\rceil-1}$ and call $S_0 = T$.
Then, for increasing levels $k>0$, create $S_k$ from $S_{k-1}$ as follows:
\begin{enumerate}
\item If $k$ is odd, find the maximal runs of (say, $t>1$ copies of) equal 
symbols $Y$ in $S_{k-1}$ such that $|exp(Y)| \le \ell_k$, create a new grammar
rule $X \rightarrow Y^t$, and replace the run by $X$. The other symbols are
copied onto $S_k$ as is.
\item If $k$ is even, generate a function $\pi_k$ that randomly reorders the 
symbols of $S_{k-1}$ and define local minima as the positions $1<i<|S_{k-1}|$ 
such that $\pi_k(S_{k-1}[i-1]) > \pi_k(S_{k-1}[i])<\pi_k(S_{k-1}[i+1])$. Place a
block boundary after each local minimum, and before and after the symbols 
$Y$ with $|exp(Y)| > \ell_k$. Create new rules for the resulting blocks of 
length more than 1 and replace them in $S_k$ by their corresponding 
nonterminals. Leave other symbols as is.
\end{enumerate}

Our plan is to extract $T$ left to right from its SLP, in $O(n)$ time, and carry out the described process in streaming form.
The only obstacle to perform the process at level $k$ in a single left-to-right
pass is the creation of the functions $\pi_k$ without knowing in advance the
alphabet of $S_{k-1}$. We can handle this by maintaining two balanced trees. 
The first, $T_{id}$, is sorted by the actual symbol identifiers, and stores for each 
symbol a pointer to its node in the second tree, $T_{pos}$. The tree $T_{pos}$ is sorted by the current
$\pi_k$ values (which evolve as new symbols arise), that is, the $\pi_k$ value of a symbol is its inorder position in $T_{pos}$. We can know the current 
value of 
a symbol in $\pi_k$ by going up from its node in $T_{pos}$ to the root, adding up one plus the number of nodes in the left subtree of the nodes we reach from their right child (so $T_{pos}$ stores
subtree sizes to enable this computation). Two symbols are then compared in
logarithmic time by 
computing their $\pi_k$ values using $T_{pos}$.

When the next symbol is not found in $T_{id}$, it is inserted in both 
trees. Its rank $r$ in $T_{pos}$ is chosen at random in $[1,|T_{id}|+1]$. We use 
the subtree sizes to find the insertion point in $T_{pos}$, starting from the root:
let $t_l$ be the size of the left child of a node. If $r \le t_l+1$ we continue
by the left child, otherwise we subtract $t_l+1$ from $r$ and continue by the
right child. The balanced tree rotations maintain the ranks of the nodes, so the tree can be 
rebalanced after the insertion adds a leaf.

Our space budget does not allow us maintaining the successive strings $S_k$.
Rather, we generate $S_0 = T$ left to right in linear time using the given SLP and have one iterator per level $k$ (the number of levels until having
a single nonterminal is logarithmic~\cite[Remark 3.16]{KNOlatin22}). 
Each time the process at some level $k-1$ produces a new 
symbol, it passes that new symbol on to the next level, $k$. When the last 
symbol of $T$ is consumed, all the levels in turn close their processes, bottom-up; 
the LCG comprises the rules produced along all levels. 

The total space used is proportional to the number of distinct symbols across all the
levels of the grammar. This can be larger than the grammar size because symbols $X$ with $|exp(X)| > \ell_k$ are not replaced in level $k$, so they exist in the next levels as well. To avoid this, we perform a twist
that ensures that every distinct grammar symbol is stored only in $O(1)$ levels. 
The twist is not to store in the trees the symbols that cannot form groups in this level, 
that is, those $X$ for which $|exp(X)| > \ell_k$. Since then the symbols stored in the 
tree for even levels $k$ are forced to form blocks (no two consecutive minima can exist),
they will no longer exist in level $k+1$. Note that the sizes of the trees used for the 
symbols at level $k$ are then proportional to the number of nonterminals of that level 
in the produced grammar. 

There is a deterministic bound $O(\delta\log\frac{n}{\delta})$ on the total 
number of nonterminals in the generated grammar~\cite[Corollary 3.12]{KNOlatin22}, and thus 
on the total sizes of the balanced
trees. Here, $\delta$ is the compressibility measure based on substring 
complexity, and size $O(\delta\log\frac{n}{\delta})$ is optimal for every
$n$ and $\delta$~\cite{KNOlatin22,KNPtit22}. The size $g_{lc}$ of the produced LCG could 
be higher, as for some choices of letter permutations on various levels some right-hand of the productions can be of not-constant length, however,
but it is still $O(\delta\log\frac{n}{\delta})$ in 
expectation and with high probability~\cite[Theorem 3.13]{KNOlatin22}. Because the sum of the lengths of the 
strings $S_k$ is $O(n)$~\cite[Corollary 3.15]{KNOlatin22}, we produce the LCG in time 
$O(n\log g_{lc})$; the $\log g_{lc}$ comes from the cost of balanced tree operations.

\begin{theorem} 
Given an SLP with $g$ rules for a text $T [1..n]$, we can build whp a LCG of size 
$g_{lc} = O(\delta\log\frac{n}{\delta})$ for $T$ within $O(n\log g_{lc})$ time and 
$O(g+g_{lc})$ space.
\end{theorem}

If we know $\delta$, we can abort the construction as soon 
as its total size exceeds $c \cdot \delta\log\frac{n}{\delta}$ for some suitable
constant $c$, and restart the process afresh. After $O(1)$  
attempts in expectation, we will obtain a locally consistent grammar of size 
$O(\delta\log\frac{n}{\delta})$~\cite[Corollary 3.15]{KNOlatin22}. The grammar we produce, in $O(n\log g_{lc})$ expected time, is then of guaranteed size $g_{lc} = O(\delta\log\frac{n}{\delta})$.

\section{Building the LZ parse from a LCG in $\tilde{O}(z)$ time}
\label{sec:lz77fromlc}

One of the many advantages of LCGs compared to general SLPs is that, related to Definition~\ref{def:slpindex}, they may allow trying out only $O(\log |P|)$ splitting positions of $P$ in order to discover all their primary occurrences, as opposed to $m-1$ if using a generic SLP. This is the case of the LCG of size $O (\delta \log \frac n \delta)$ of the previous section \cite{KNOlatin22}, and of a generalization of it \cite{christiansen2020optimal}. The second grammar is a generalization of the first in the sense that any grammar produced with the first method \cite{KNOlatin22} can be produced by the second \cite{christiansen2020optimal}, and therefore every property we prove for the second method holds for the first as well. The first method introduces a restriction to produce grammars of size $O(\delta\log\frac{n}{\delta})$, whereas the second kind has a weaker space bound of $O(\gamma\log\frac{n}{\gamma})$, where $\gamma \ge \delta$ is the size of the smallest string attractor of $T$ \cite{kempa2018roots} (concretely, the parsing is as in Section~\ref{sec:locally_consistent_parsing} but does not enforce the condition $\exp(X) \le \ell_k$).
We now show how this feature enables us to find the LZ parse of those LCGs in time $O(z\log^4 n)$. We will then stick to the more general LCG \cite{christiansen2020optimal}; the results hold for the other too \cite{KNOlatin22}, as explained. 

Our technique combines results used for Theorems~\ref{thm:npolylogn} and \ref{thm:SLP_to_LZ77}: we will use exponential search, as in Section~\ref{sec:lz77gzlogn} to find the next phrase $T[i..k]$, and will use the data structures of Section~\ref{sec:lz77npolylog} to search for its leftmost occurrence in $T$; the fact that we will need to check just a logarithmic amount of splitting points will yield the bound.
We start with an analogue of Lemma \ref{lem:basic_checks} for our LCG; we get better bounds in this case.

\begin{lemma}
    \label{lem:basic_checks_lc}
    Given the LCG \cite{christiansen2020optimal} of size $g_{lc}$ of $T[1..n]$, we can build in $O(g_{lc})$ time and space a data structure supporting the same operations listed in Lemma~\ref{lem:basic_checks}, all in $O(\log n)$ time.
\end{lemma}

\begin{proof}
    Since the LCG grammar is already balanced, accessing $T[i]$ in $O(\log n)$ time is immediate. The Karp-Rabin fingerprints can be computed with the structure of  Christiansen et al.~\cite[Thm.~A.3]{christiansen2020optimal},
    which can be built in $O (g_{lc})$ space and time. 
    To compute longest common prefixes we use the approach of Kempa and Kociumaka \cite[Thm.~5.25]{kempa2023collapsing}, which works on our particular LCG without  any extra structure. If we modify it so that we also look for the first mismatching character, this gives us the (co)lexicographic order within the same time.\qed
\end{proof}

Let us consider the cost to build the data structures of Section~\ref{sec:lz77npolylog}. Using Lemma \ref{lem:basic_checks_lc}, we can sort the multisets $\mathcal{Y}$ and $\mathcal{Z}$ in time $O(g_{lc} \log g_{lc} \cdot \log n)$. This time dominates the construction time of the z-fast tries for $\mathcal{Y}$ and $\mathcal{Z}$, the grid structure $\mathcal G$, and the two-dimensional range minimum query mentioned in Lemma~\ref{lem:primary}. Further, because $g_{lc} \le \gamma\log\frac{n}{\gamma}$ \cite{christiansen2020optimal} and $\gamma \le z$ \cite{kempa2018roots}, this time is in $O(z\log^2 n\log(n/z))$.

After building those components, we start parsing the text using the exponential search of Section~\ref{sec:lz77gzlogn}.
To test whether the candidate phrase $T[i..j]$ occurs starting to the left of $i$, we use the LCG search algorithm for $T[i..j]$ provided by the LCG. 
Christiansen et al.~\cite{christiansen2020optimal} observed that we need to check only $O(\log(j-i))$ splitting points to find every primary occurrence of $T[i..j]$. They find the splitting points through a linear-time parse of $T[i..j]$, but we can do better by reusing the locally consistent parsing used to build the LCG. While we do not store the strings $S_k$ of Section~\ref{sec:locally_consistent_parsing}, we can recover the pieces that cover $T[i..j]$ by traversing the (virtual) grammar tree from the root towards that substring of $T$. 

\begin{lemma}[\cite{christiansen2020optimal}] \label{lem:M}
Let $M_0(i,j) = \{ i,j-1 \}$. For any $k>0$, let $M_k(i,j)$ contain the first and last positions ending a block of $S_k$ that are within $T[i..j-1]$ but do not belong to $M_{k'}(i,j)$ for any $k'<k$. Then, $M(i,j) = \cup_k M_k(i,j)$ is of size $O(\log(j-i))$ and the splitting point of every primary occurrence of $T[i..j]$ in $T$ belongs to $M(i,j)$.   
\end{lemma}
\begin{proof}
Our definition of $M(i,j)$ includes the positions in Definitions 4.7 and 4.8 of Christiansen et al.~\cite{christiansen2020optimal} (they use $B_r$ and $\hat{B}_r$ instead of our even and odd levels $S_k$). The property we state corresponds to their Lemma~6.4 \cite{christiansen2020optimal}.\qed
\end{proof}
To compute $M(i,j)$, then, we descend from the root of the (virtual) parse tree of the LCG towards the lowest nonterminal $X$ that fully contains $T[i..j]$, and continue from $X$ towards the leaf $L$ that contains $T[i]$. We then start adding to $M(i,j)$ the endpoint of $L$ (which is $i$), and climb up to its parent $P$. If $P$ ends in the same position of $L$, we shift $P$ to its next sibling. We now set $L = P$, add the last position of $L$ to $M(i,j)$, climb up to its parent $P$, and so on until the last position of $L$ exceeds $j-1$ (which may occur when reaching $X$ or earlier). We proceed analogously with the path from $X$ to the leaf that contains $T[j-1]$.

We visit $O(\log n)$ nodes in this process, but since the LCG may not be binary, we may need $O(\log n)$ time to find the proper children of a node. The total time is then $O(\log^2 n)$.


Once the set $M(i,j)$ of splitting points is found, we search for each of them as in Lemma~\ref{lem:primary}, each in time $O(\log^2 n)$. Therefore, the total time to check a candidate $T[i..j]$ is $O(\log^2 n \log(j-i))$. In turn, the exponential search that finds the next phrase $T[i..k]$ carries out $O(\log(k-i))$ such checks, with $j-i \le 2(k-i)$, thus the total time to find the next phrase is $O(\log^2 n \log^2 (k-i))$. By Jensen's inequality again, this adds up to $O(z\log^2 n \log^2 (n/z))$, which absorbs the time to build the data structures.

\begin{theorem}
\label{thm:lcpolylog}
Given the LCG of Christiansen et al.~\cite{christiansen2020optimal} of size $g_{lc}$ of $T[1..n]$, we can build whp the LZ parse of $T$ in $O(z \log^2 n \log^2 (n/z))$ time and $O(g_{lc})$ extra space. The result also holds verbatim for the LCG of Kociumaka et al.~\cite{KNOlatin22}.
\end{theorem}

\section{Conclusions}

We have contributed to the problem of {\em compressed conversions}, that is, using asymptotically optimal space, between various compression formats for repetitive data. Such a space means linear in the input plus output size, which outrules the possibility of decompressing the data. This is crucial to face the sharp rise the size of data in sequence form has experienced in the last decades, which requires manipulating the data always in compressed form. To the best of our knowledge, we are the first to propose methods to build the Lempel-Ziv parse of a text directly from its straight-line program representation. Our methods work in time $O(n\log^2 n)$ and $O(gz\log^2 n)$. The second is polynomial on the size of the compressed data and we thus call it a {\em fully-compressed conversion}; such methods can be considerably faster when the data is highly compressible. We also gave methods to convert from straight-line programs to locally consistent grammars, which enable faster and more complex queries, in $O(n\log n)$ time. As a showcase for their improved search capabilities, we show how to produce the Lempel-Ziv parse from those grammars in time $O(z\log^4 n)$, another fully-compressed conversion. All of our conversions work with high probability.  

Obvious open problems are obtaining better running times without using more space. 
Furthermore, we think that approaches similar to those described in this article can be applied to effectively compute other parses, such as the lexicographic parse \cite{navarro2020approximation}. We plan to address these in the extended version.

\bibliographystyle{plain}
\bibliography{main}

\begin{thebibliography}{10}

\bibitem{arimura2023optimally}
Hiroki Arimura, Shunsuke Inenaga, Yasuaki Kobayashi, Yuto Nakashima, and Mizuki
  Sue.
\newblock Optimally computing compressed indexing arrays based on the compact
  directed acyclic word graph.
\newblock In {\em International Symposium on String Processing and Information
  Retrieval}, pages 28--34. Springer, 2023.

\bibitem{BBPV10}
Djamal Belazzougui, Paolo Boldi, Rasmus Pagh, and Sebastiano Vigna.
\newblock Fast prefix search in little space, with applications.
\newblock In {\em Proc. 18th Annual European Symposium on Algorithms (ESA),
  Part {I}}, pages 427--438, 2010.

\bibitem{BLRSRW15}
P.~Bille, G.~M. Landau, R.~Raman, K.~Sadakane, S.~S. Rao, and O.~Weimann.
\newblock Random access to grammar-compressed strings and trees.
\newblock {\em SIAM Journal on Computing}, 44(3):513--539, 2015.

\bibitem{BGCSVV17}
Philip Bille, Inge~Li G{\o}rtz, Patrick~Hagge Cording, Benjamin Sach,
  Hjalte~Wedel Vildh{\o}j, and S{\o}ren Vind.
\newblock Fingerprints in compressed strings.
\newblock {\em J. Comput. Syst. Sci.}, 86:171--180, 2017.

\bibitem{BW94}
M.~Burrows and D.~Wheeler.
\newblock A block sorting lossless data compression algorithm.
\newblock Technical Report 124, Digital Equipment Corporation, 1994.

\bibitem{CLLPPSS05}
M.~Charikar, E.~Lehman, D.~Liu, R.~Panigrahy, M.~Prabhakaran, A.~Sahai, and
  A.~Shelat.
\newblock The smallest grammar problem.
\newblock {\em IEEE Transactions on Information Theory}, 51(7):2554--2576,
  2005.

\bibitem{Cha88}
B.~Chazelle.
\newblock A functional approach to data structures and its use in
  multidimensional searching.
\newblock {\em SIAM Journal on Computing}, 17(3):427--462, 1988.

\bibitem{CEKNPtalg20}
A.~R. Christiansen, M.~B. Ettienne, T.~Kociumaka, G.~Navarro, and N.~Prezza.
\newblock Optimal-time dictionary-compressed indexes.
\newblock {\em ACM Transactions on Algorithms}, 17(1):article 8, 2020.

\bibitem{christiansen2020optimal}
Anders~Roy Christiansen, Mikko~Berggren Ettienne, Tomasz Kociumaka, Gonzalo
  Navarro, and Nicola Prezza.
\newblock Optimal-time dictionary-compressed indexes.
\newblock {\em ACM Transactions on Algorithms}, 17(1):1--39, 2020.

\bibitem{CNPjcss21}
F.~Claude, G.~Navarro, and A.~Pacheco.
\newblock Grammar-compressed indexes with logarithmic search time.
\newblock {\em Journal of Computer and System Sciences}, 118:53--74, 2021.

\bibitem{DAB15}
Richard~M. Durbin, Adam Auton, and Lisa~D. Brooks.
\newblock A global reference for human genetic variation.
\newblock {\em Nature}, 526(7571):68--74, 2015.

\bibitem{GNPjacm19}
T.~Gagie, G.~Navarro, and N.~Prezza.
\newblock Fully-functional suffix trees and optimal text searching in
  {BWT}-runs bounded space.
\newblock {\em Journal of the ACM}, 67(1):article 2, 2020.

\bibitem{GJL21}
Moses Ganardi, Artur Jez, and Markus Lohrey.
\newblock Balancing straight-line programs.
\newblock {\em J. {ACM}}, 68(4):27:1--27:40, 2021.

\bibitem{gawrychowski2011pattern}
Pawe{\l} Gawrychowski.
\newblock Pattern matching in {L}empel-{Z}iv compressed strings: fast, simple,
  and deterministic.
\newblock In {\em European Symposium on Algorithms}, pages 421--432. Springer,
  2011.

\bibitem{gawrychowski2018optimal}
Pawe{\l} Gawrychowski, Adam Karczmarz, Tomasz Kociumaka, Jakub {\L}{\k{a}}cki,
  and Piotr Sankowski.
\newblock Optimal dynamic strings.
\newblock In {\em Proceedings of the Twenty-Ninth Annual ACM-SIAM Symposium on
  Discrete Algorithms}, pages 1509--1528. SIAM, 2018.

\bibitem{jez2015faster}
Artur Je{\.z}.
\newblock Faster fully compressed pattern matching by recompression.
\newblock {\em ACM Transactions on Algorithms (TALG)}, 11(3):1--43, 2015.

\bibitem{karp1987efficient}
Richard~M Karp and Michael~O Rabin.
\newblock Efficient randomized pattern-matching algorithms.
\newblock {\em IBM journal of research and development}, 31(2):249--260, 1987.

\bibitem{KK20}
D.~Kempa and T.~Kociumaka.
\newblock Resolution of the {Burrows-Wheeler Transform} conjecture.
\newblock In {\em Proc. 61st {IEEE} Annual Symposium on Foundations of Computer
  Science (FOCS)}, pages 1002--1013, 2020.

\bibitem{kempa2023collapsing}
Dominik Kempa and Tomasz Kociumaka.
\newblock Collapsing the hierarchy of compressed data structures: Suffix arrays
  in optimal compressed space, 2023.

\bibitem{KK17}
Dominik Kempa and Dmitry Kosolobov.
\newblock {LZ-End} parsing in compressed space.
\newblock In {\em Proc. 27th Data Compression Conference (DCC)}, pages
  350--359, 2017.

\bibitem{kempa2018roots}
Dominik Kempa and Nicola Prezza.
\newblock At the roots of dictionary compression: string attractors.
\newblock In {\em Proceedings of the 50th Annual ACM SIGACT Symposium on Theory
  of Computing}, pages 827--840, 2018.

\bibitem{KY00}
J.~C. Kieffer and E.-H. Yang.
\newblock Grammar-based codes: {A} new class of universal lossless source
  codes.
\newblock {\em IEEE Transactions on Information Theory}, 46(3):737--754, 2000.

\bibitem{KNOlatin22}
T.~Kociumaka, G.~Navarro, and F.~Olivares.
\newblock Near-optimal search time in {$\delta$}-optimal space.
\newblock In {\em {Proc. 15th Latin American Symposium on Theoretical
  Informatics (LATIN)}}, 2022.

\bibitem{KNPtit22}
T.~Kociumaka, G.~Navarro, and N.~Prezza.
\newblock Towards a definitive measure of repetitiveness.
\newblock {\em IEEE Transactions on Information Theory}, 2022.
\newblock To appear. See http://arxiv.org/abs/1910.02151.

\bibitem{LZ76}
A.~Lempel and J.~Ziv.
\newblock On the complexity of finite sequences.
\newblock {\em IEEE Transactions on Information Theory}, 22(1):75--81, 1976.

\bibitem{Navacmcs20.3}
G.~Navarro.
\newblock Indexing highly repetitive string collections, part {I}:
  Repetitiveness measures.
\newblock {\em ACM Computing Surveys}, 54(2):article 29, 2021.

\bibitem{NP18}
G.~Navarro and N.~Prezza.
\newblock Universal compressed text indexing.
\newblock {\em Theoretical Computer Science}, 762:41--50, 2019.

\bibitem{navarro2014wavelet}
Gonzalo Navarro.
\newblock Wavelet trees for all.
\newblock {\em Journal of Discrete Algorithms}, 25:2--20, 2014.

\bibitem{navarro2021indexing}
Gonzalo Navarro.
\newblock Indexing highly repetitive string collections, part i: repetitiveness
  measures.
\newblock {\em ACM Computing Surveys (CSUR)}, 54(2):1--31, 2021.

\bibitem{navarro2020approximation}
Gonzalo Navarro, Carlos Ochoa, and Nicola Prezza.
\newblock On the approximation ratio of ordered parsings.
\newblock {\em IEEE Transactions on Information Theory}, 67(2):1008--1026,
  2020.

\bibitem{PP17}
Alberto Policriti and Nicola Prezza.
\newblock From {LZ77} to the run-length encoded {B}urrows-{W}heeler
  {T}ransform, and back.
\newblock In {\em Proc. 28th Annual Symposium on Combinatorial Pattern Matching
  (CPM)}, volume~78 of {\em LIPIcs}, pages 17:1--17:10, 2017.

\bibitem{Ryt03}
W.~Rytter.
\newblock Application of {L}empel-{Z}iv factorization to the approximation of
  grammar-based compression.
\newblock {\em Theoretical Computer Science}, 302(1-3):211--222, 2003.

\end{thebibliography}

\end{document}